\documentclass[12pt]{article}

\usepackage[margin=1in]{geometry}
\usepackage{amsmath,amssymb,amsfonts}
\usepackage{amsthm}
\usepackage{graphicx}
\usepackage{booktabs}
\usepackage[authoryear,round]{natbib}
\usepackage{xcolor}
\usepackage{float}
\usepackage{makecell}
\usepackage{enumitem}
\usepackage{setspace}
\usepackage{xurl}
\usepackage{hyperref}
\hypersetup{hidelinks}
\usepackage{amsmath}
\usepackage{authblk}
\newtheorem{assumption}{Assumption}
\newtheorem{theorem}{Theorem}

\def\cW{\mathcal{W}}
\def\cV{\mathcal{V}}
\def\cR{\mathcal{R}}
\def\cA{\mathcal{A}}

\def\cM{\mathcal{M}}
\def\cY{\mathcal{Y}}

\DeclareMathOperator*{\argmax}{arg\,max}

\def\bQ{\bar{Q}}

\newcommand{\tauv}{\tau^{v}}
\newcommand{\tauvr}{\tau^{v,r}}

\newcommand{\drefv}{d_{\mathrm{ref}}^{v}}
\newcommand{\drefvr}{d_{\mathrm{ref}}^{v,r}}

\newcommand{\doptv}{d_{\mathrm{opt}}^{v}}
\newcommand{\doptvr}{d_{\mathrm{opt}}^{v,r}}
\newcommand{\doptvn}{d_{\mathrm{opt,n}}^{v}}
\newcommand{\doptvrn}{d_{\mathrm{opt,n}}^{v,r}}

\newcommand{\donev}{d_{1}^{v}}
\newcommand{\dtwovr}{d_{2}^{v,r}}

\newcommand{\donevn}{d_{1,n}^{v}}
\newcommand{\dtwovrn}{d_{2,n}^{v,r}}

\newcommand{\donevnminusj}{d_{1,n,j}^{v}}
\newcommand{\dtwovrnminusj}{d_{2,n,j}^{v,r}}

\def\indep{\perp\!\!\!\perp}
\newcommand{\logit}{\mathop{\mathrm{logit}}}

\def\dto{\overset{d}{\to}}

\definecolor{vz}{rgb}{0.0, 0.0, 0.8}

\newcommand{\Val}{\mathsf{Val}}
\newcommand{\Train}{\mathsf{Train}}

\newcommand{\Var}{Var}

\title{
\makebox[\textwidth][c]{%
\begin{minipage}{1.02\textwidth}
\centering
A Causal Inference Approach for Evaluating 
Diagnostic Tests and
AI-Enabled Medical Devices:
From Effect
Modification to
Information-Augmented Decision-Making
\end{minipage}%
}
}

\author[]{Wenxin Zhang}
\author[]{Rachael Phillips}
\author[]{Mark van der Laan \footnote{We would like to extend sincere gratitude to Gene Pennello of the US Food and Drug Administration for his valuable contribution to this project.}
}

\affil[]{Division of Biostatistics, University of California, Berkeley}

\date{}

\begin{document}

\maketitle

% \newpage
\vspace{-2em}
\begin{abstract}
Diagnostic medical tests and devices provide useful information for evaluating the potential benefits and risks of therapeutic treatments. However, unlike treatments, their impact on health outcomes is generally indirect, because measuring diagnostic information typically does not itself affect patient outcomes. This indirect pathway makes it challenging to define and evaluate the clinical utility of medical tests and devices. In this paper, we develop a causal inference approach for evaluating diagnostic tests by distinguishing between explanatory and pragmatic effectiveness. 
Explanatory effectiveness evaluates whether a diagnostic test result provides treatment-relevant information beyond baseline covariates, addressing the scientific question of whether and how the test result explains additional treatment-effect heterogeneity, which we characterize using a variance-based treatment-effect variable importance measure.
Pragmatic effectiveness evaluates whether incorporating this information into personalized treatment decisions improves expected outcomes, a question central to decision-making for patients, clinicians, and other health care stakeholders. We formalize pragmatic effectiveness as a contrast between the expected outcomes of optimal personalized treatment rules defined on different information sets, with or without access to the diagnostic test result. We establish identification of the proposed estimand and provide Targeted Maximum Likelihood Estimation (TMLE) and cross-validated TMLE procedures for statistical estimation and inference without parametric assumptions. Simulation studies and a synthetic colorectal cancer application illustrate the interpretation of the proposed estimands and the estimation performance. 
More broadly, this approach provides a causal inference perspective for understanding and
evaluating AI-enabled devices by distinguishing their dual roles as
information-enrichment tools and personalized decision-optimization tools, clarifying whether AI improves outcomes by expanding the information available
for downstream decisions, by improving the decision rule used to act on that information, or through both pathways.
\end{abstract}

\setstretch{1.5}
\section{Introduction}

Artificial intelligence / machine learning (AI/ML)-enabled medical tools are increasingly being embedded into clinical workflows. Many of these tools have been developed into commercial medical devices or diagnostic tests for a broad range of clinical applications, including, for example, disease detection and diagnosis, patient monitoring and risk prediction, and treatment guidance or dosing \citep{singh2025ai}. 
Many AI/ML-enabled medical devices consist entirely of software and are referred to as Software as a Medical Device (SaMD; see FDA draft guidance, \citeyear{fda2025samd}, not for implementation), e.g., computer-assisted detection or image segmentation \citep{petrick2023regulatory}. 
As AI/ML moves from model development to medical device decision support, central evaluation questions include not only whether the device predicts accurately or performs well against standard benchmarks, but its effect on clinical care \citep{fda2024transparency,shick2024transparency}.
This clinical impact is not only central to patients and clinicians, but may also be important to broader health-care stakeholders, especially payers \citep{schulman2010policy,fda2026cms}.

The core functionality and intent of use of decision support medical tools shape the appropriate evaluation target(s) and the statistical approach needed to generate evidence for clinical utility from, e.g., explanatory clinical trials, pragmatic studies, or real-world data. For action-oriented tools that support treatment selection (or more general management options), evaluation is often anchored to comparisons of the treatment decisions or strategies that the tool directly suggests or assists with. In contrast, the evaluation of information-oriented tools is more challenging, because diagnostic information is not itself a treatment, and measuring such information generally does not directly affect health outcomes.

This challenge is central to the broader evaluation problem for diagnostic medical tests and devices. Like many AI-enabled information-oriented tools, diagnostic tests and devices are intended to measure biomarkers or generate diagnostic information for detecting, diagnosing, and monitoring disease, and for guiding improved healthcare. For example, predictive biomarkers, typically measured by in vitro diagnostic assays or SaMDs, are used to identify clinically relevant sub-groups of individuals whose expected benefit or risk from a medical product or exposure may differ from that of otherwise similar individuals without the biomarker \citep{FDA-NIH2016BEST}. 
As a special case, companion diagnostic tests, which are in vitro diagnostic devices essential for the safe and effective use of a corresponding therapeutic product (\citealp{FDA2014Companion,FDA2020Developing};
see also FDA draft guidance, \citeyear{FDA2016Principles},
not for implementation), may measure biomarker(s) that is (are) predictive of a therapeutic product’s safety or efficacy, often with the intent to use the biomarker test measurement(s) to determine eligibility for that product. 

While existing literature has compared different trial designs for evaluating such medical tests and devices, 
evaluating the utility of diagnostic tests and devices requires careful specification of the target estimand in relation to the scientific or decision-making question of interest \citep{bossuyt2012beyond,pennello2016clinical,zhou2026statistical}.
A common starting point in the diagnostic and biomarker literature is effect modification, that is, whether a test result is predictive of treatment effect. While scientifically important, this perspective does not by itself establish clinical utility. Information that refines treatment effect heterogeneity does not improve outcomes if treatment decisions remain unchanged or if the same treatment is optimal regardless of the test result. 

This distinction supports the development of a causal inference approach for evaluating diagnostic information along two related but distinct dimensions: its ability to explain treatment-effect heterogeneity and its impact on expected outcomes through the downstream treatment decisions it informs. We refer to these dimensions as \emph{explanatory effectiveness} and \emph{pragmatic effectiveness}, respectively. Explanatory effectiveness evaluates whether a diagnostic test result provides relevant treatment-efficacy information beyond a prespecified reference information set, such as baseline covariates. We formalize this notion by comparing conditional average treatment effect functions defined under the reference information set and under the enriched information set that additionally includes the diagnostic test result. This yields a variance-based treatment-effect variable-importance estimand \citep{hines2025variable,li2023targeted,levy2021fundamental} that quantifies the residual variation in treatment-effect heterogeneity explained by the diagnostic result beyond reference information. 
Pragmatic effectiveness evaluates whether using that information in treatment decision-making has an effect on expected outcome.
In this sense, explanatory effectiveness characterizes the information content of a diagnostic signal, whereas pragmatic effectiveness characterizes the outcome value of acting on that information.

To formalize pragmatic effectiveness, we compare the expected outcomes of optimal personalized treatment rules
defined under different information sets, with and without access to the
diagnostic test result, building upon the statistical literature on identifying and optimizing dynamic treatment regimes/rules \citep{robins1986new,murphy2003optimal,chakraborty2014dynamic,van2015targeted,luedtke2016super}. This formulation makes explicit the indirect pathway
from diagnostic information to patient outcomes and allows the evaluation of user-prespecified or data-adaptively learned rules with access to medical test or device information. Methodologically, we establish
identification of the proposed pragmatic effectiveness estimands under standard
causal assumptions and provide a Targeted Maximum Likelihood Estimation (TMLE)
approach for nonparametric estimation and inference. Since the optimal treatment
rules under comparison may themselves be estimated from data, we further provide a cross-validated TMLE
(CV-TMLE) procedure that separates rule learning from value evaluation
\citep{zheng2011cross,van2015targeted,luedtke2016super}. We study finite-sample
performance of the proposed approach through simulations and illustrate the approach in a synthetic dataset in the context of 
a colorectal cancer study.

Importantly, our approach bridges two interconnected roles of AI/ML and presents their distinct pathways to support clinical decision-making: using AI/ML as information-generating tools that augment the information available for treatment decisions, and using AI/ML as treatment-optimization tools that improve how available information is translated into personalized treatment rules, or both. This general methodology may provide a basis for application across areas such as health monitoring, clinical care delivery, personalized medicine, AI-deployment assessment, and broader decision-making contexts in which information augmentation and decision-rule optimization can jointly improve target outcomes.

\section{Statistical setup}
\label{section:setup}

We define the full data structure as
\[
X = \left(W, R, Y(1), Y(0)\right) \sim P_X \in \cM^\mathrm{full},
\]
drawn from a full-data distribution $P_X$ within a nonparametric model $\cM^\mathrm{full}$. Here, $W \in \cW$ denotes baseline covariates, $A \in \cA$ is treatment, $R \in \cR$ represents the underlying test result, and $Y(a) \in \cY$ is the counterfactual outcome under a treatment level $a \in \cA$, with higher values indicating better health outcomes. Our approach allows $W$ and $R$ to be multivariate, and $A$ to be a multi-level categorical variable.
Without loss of generality, we consider a binary treatment $\cA=\{0,1\}$ and a bounded outcome space $\cY=[0,1]$ for the rest of the paper.

The observed data are
$O=(W,R,A,Y)\sim P\in\cM$,
drawn from a distribution $P$ within a nonparametric model $\cM$. 
We suppose the following nonparametric structural causal model (SCM) \citep{pearl2009causality} to generate the observed data:
\begin{align*}
    W,R & = f_{W,R}\left(U_{W,R}\right),         \\
    % T & = f_{T}\left(W, U_{T}\right),       \\
    A & = f_{A}\left(W, R, U_{A}\right), \\
    Y & = f_{Y}\left(W, R, A, U_{Y}\right),
\end{align*}
where $U_{W,R},U_A,U_Y$ are independent exogenous noise variables.

\section{Explanatory effectiveness}
Let \(V=V(W)\in\cV\) denote a prespecified subvector or summary of baseline covariates \(W\). We refer to $V$ as reference information of interest, and to $(V,R)$ as the enriched information set obtained by augmenting $V$ with diagnostic test result $R$. We allow \(V\) to be empty, in which case quantities conditional on \(V\) are interpreted as their marginal counterparts.

Explanatory effectiveness characterizes whether a diagnostic test result $R$ modifies treatment effectiveness beyond $V$. In general, there are multiple ways to summarize such information, including contrasts in average treatment effects across test-defined subgroups or model-based summaries of effect modification. In this paper, we focus on a formulation defined by conditional average treatment effect (CATE), where the goal is to assess whether $R$ further shifts the CATE function beyond $V$.

Specifically, we define the baseline CATE given $V$ as
\[
\tauv(V)
:=
E_{P_X}\!\left[Y(1)-Y(0) \mid V \right].
\]
We define the enriched-information CATE with the diagnostic test result $R$ as
\[
\tauvr(V,R)
:=
E_{P_X}\!\left[ Y(1)-Y(0) \mid V,R \right].
\]
The difference $\tauvr(V,R) - \tauv(V)$ is the discrepancy between the two CATE functions with and without $R$ given $V$.
A natural summary of this discrepancy is the mean squared difference $E\left[
\left\{
\tauvr(V,R)-\tauv(V)
\right\}^{2}
\right]$, which measures the expected squared difference in treatment-effect prediction when the information set is enriched from $V$ to $(V,R)$.
Since
$
\tauv(V)=E[\tauvr(V,R)\mid V]$,
this quantity also admits the equivalent representation
\[
E_{P_X}\left[
\left\{
\tauvr(V,R)-\tauv(V)
\right\}^{2}
\right]= E_{P_X}\left[
\Var\big(\tauvr(V,R)\mid V\big)
\right] := \Psi^{\mathrm{exp,vim}}(P_X),
\]
which corresponds to the variance-based treatment effect variable importance measure (VIM)
\citep{hines2025variable,li2023targeted,levy2021fundamental}, quantifying the average residual variation in CATE explained by the diagnostic test result beyond reference information. A value of zero implies that the test result does not explain additional treatment effect heterogeneity beyond reference information, while positive values indicate that the diagnostic test result modifies the CATE function. Estimation and inference for this VIM can be performed using
estimating-equation estimators \citep{hines2025variable} or targeted
maximum likelihood estimation \citep{li2023targeted}.

\section{Pragmatic effectiveness}

Although explanatory effectiveness provides a causal characterization of whether a diagnostic test result explains treatment-effect heterogeneity, such evidence alone does not imply that use of the diagnostic test or device will improve clinical outcomes. Diagnostic information can be predictive of treatment effectiveness, yet have no pragmatic value if it does not alter the treatment decisions actually made. For example, the treatment strategy based on reference information $V$ alone may assign the same treatment as the rule based on the enriched information set $(V,R)$, or clinicians may already prescribe near-optimal treatments using $V$ without relying on $R$. In these cases, even a diagnostically informative test may fail to improve patient outcomes because the additional information does not translate into a different or better treatment decision.

These considerations motivate evaluating diagnostic tests and devices not only by whether they explain treatment-effect heterogeneity, but also by whether their use improves downstream decision-making and clinical outcomes. We formalize this latter notion through a causal estimand of pragmatic effectiveness.

\subsection{Causal estimand}

We evaluate the pragmatic impact of the diagnostic test by comparing treatment rules that are defined on different information sets: a rule that uses only the reference information and a rule that is additionally allowed to use the diagnostic test result. Throughout, we use the term treatment rule to refer to a map from the available information to a treatment decision. 
This formulation builds on a treatment-rule perspective in the dynamic treatment regime (rule) literature \citep{robins1986new,murphy2003optimal,chakraborty2014dynamic,van2015targeted} and is related to decision-theoretic approaches for evaluating biomarker value in treatment selection \citep{janes2014approach,huang2015characterizing,zhou2021net}.

Let $\donev:\cV \to\cA$ and $\dtwovr:\cV \times\cR\to\cA$
denote a treatment rule based on reference information $V$
alone and a treatment rule based on the enriched information set
\((V,R)\), respectively. 
To define the causal estimand, we consider two modified SCMs:
\[
\begin{array}{l@{\qquad\qquad}l}
\textbf{Modified SCM 1}
&
\textbf{Modified SCM 2}\\[0.5em]
\begin{aligned}
    & W,R = f_{W,R}\left(U_{W,R}\right),\\
    % & T^* = 1,\\
    & A^* = \dtwovr\left(V, R\right),\\
    & Y^*_{\dtwovr} = f_{Y}\left(W, R, A^*, U_{Y}\right),
\end{aligned}
&
\begin{aligned}
    & W,R = f_{W,R}\left(U_{W,R}\right),\\
    % & T^{\dagger} = 1, \\
    & A^{\dagger} = \donev\left(V\right),\\
    & Y^{\dagger}_{\donev} = f_{Y}\left(W, R, A^{\dagger}, U_{Y}\right).
\end{aligned}
\end{array}
\]

For any treatment rule $d$, let $P_d$ denote the post-intervention full-data distribution induced by assigning treatment by rule $d$. We define the value of $d$ as
$
\theta(d)
:=
E_{P_d}\left[Y^d\right]$,
the mean counterfactual outcome under the intervention that assigns
treatment according to $d$. 
Define $P_{\donev}$ and $P_{\dtwovr}$ as the counterfactual distribution of data under the two SCMs, respectively. 
Denote  
$\theta(\dtwovr)
:=
E_{P_{\dtwovr}}\left[Y^{*}_{\dtwovr}\right]$ and $\theta(\donev)
:=
E_{P_{\donev}}\left[Y^{\dagger}_{\donev}\right]$.

The pragmatic effectiveness of the diagnostic test, comparing an
enriched-information rule $\dtwovr$ with a reference-information-only rule
$\donev$, is defined as
\[
\Psi^{\mathrm{prag}}_{\donev,\dtwovr}(P_X)
:=
\theta(\dtwovr)-\theta(\donev).
\]
This estimand measures the change in expected outcome achieved by allowing
individualized treatment rules to use the diagnostic test result $R$ in addition to reference information $V$, relative to assigning treatment according to a rule based on $V$ alone.

\paragraph{Pragmatic Effectiveness under Optimal Treatment Rules.}
While the general contrast \(\Psi^{\mathrm{prag}}_{\donev,\dtwovr}(P_X)\) applies to any pair of rules defined on the reference and enriched information sets, our primary causal estimand quantifies the pragmatic effectiveness of diagnostic information when treatment decisions are optimally personalized under each information set. 
Specifically, let
\[
\doptv
:=
\argmax_{d^v:\cV\to\cA} \,
\theta(d^v),
\qquad
\doptvr
:=
\argmax_{d^{v,r}:\cV \times \cR \to\cA}
\theta(d^{v,r})
\]
denote the optimal treatment rules (OTR) based on the reference information \(V\)
alone and on the enriched information set \((V,R)\), respectively. 
Our primary estimand is then defined as
\[
\Psi^{\mathrm{prag}}_{\doptv,\doptvr}(P_X)
:=
\theta(\doptvr)-\theta(\doptv).
\]
This estimand asks how much the best achievable expected outcome improves when the information available for generating an optimal personalized treatment rule is enriched from \(V\) to \((V,R)\), therefore evaluating the pragmatic value of diagnostic information beyond what is possible with reference information alone.
Because any treatment rule based on \(V\) can also be represented as a rule
based on \((V,R)\) that ignores \(R\), the reference-information rule class is
nested within the enriched-information rule class. 

The identification and estimation results below are developed for a generic pair of rules \((\donev,\dtwovr)\), which includes the primary optimal-rule contrast as the special case \((\donev,\dtwovr)=(\doptv,\doptvr)\). When the rules are estimated from data, we write the learned rules as \((\donevn,\dtwovrn)\), where the subscript $n$ indicates that they are learned from the observed sample of size $n$. For data-adaptive estimation of the contrast of optimal rules, we take \((\donevn,\dtwovrn)=(\doptvn,\doptvrn)\).

\subsection{Identification}
The identification of $\Psi_{\donev,\dtwovr}^{\mathrm{prag}}(P_X)$ from the observed data relies on the following assumptions.

\begin{assumption}[Exchangeability]
\label{assump:exchangeability}
\(\forall a\in\cA\),
$Y(a)\indep A \mid W,R$.
\end{assumption}

\begin{assumption}[Positivity]
\label{assump:positivity}
\(\forall a \in \cA, w \in \cW ,r \in \cR\),
$
P(A=a\mid W=w,R=r)>0$.
\end{assumption}

$\Psi_{\donev,\dtwovr}^{\mathrm{prag}}(P_X)$ can be identified by 
\begin{align*}
\Psi^{\mathrm{prag}}_{\donev,\dtwovr}(P_X)
&= E_{P_{\dtwovr}}\left[Y^{*}_{\dtwovr}\right] -
E_{P_{\donev}}\left[Y^{\dagger}_{\donev}\right] \\
&= E_{W,R} \, E\left[Y^{*}_{\dtwovr} \mid  W,R \right] -
 E_{W,R} \, E\left[Y^{\dagger}_{\donev} \mid  W,R \right] \\
&=  E_{W,R}\left[ E\left[Y \mid A = \dtwovr(V,R), W,R\right] - E\left[Y \mid A = \donev(V), W,R\right]  \right]
\end{align*}
Let $Q_{W,R}$ denote the marginal distribution of $(W,R)$ under $P$,
and let $Q_Y$ denote the conditional distribution of $Y$ given $(A,W,R)$,
with
$
\bQ(a,w,r)
:=
E_P\left[Y \mid A=a,W=w,R=r\right].
$
Define
$
Q(P):=(Q_{W,R},Q_Y).
$
The identified target estimand can
equivalently be written as
\[
\psi_{\donev,\dtwovr}^{\mathrm{prag}}\left(Q(P)\right)
=
\int
\left[
\bQ\left(\dtwovr(v,r),w,r\right)
-
\bQ\left(\donev(v),w,r\right)
\right]
\,dQ_{W,R}(w,r).
\]
We henceforth write the estimand simply as
$\psi_{\donev,\dtwovr}^{\mathrm{prag}}(Q)$.

Since
$
\bQ\left(\dtwovr(V,R),W,R\right)
-
\bQ\left(\donev(V),W,R\right)
=0$ whenever $\dtwovr(V,R)=\donev(V)$, the target estimand can equivalently be written as
\[
\psi_{\donev,\dtwovr}^{\mathrm{prag}}(Q)
=
E_{W,R}\left[
I\left(\dtwovr(V,R)\neq\donev(V)\right)
\left\{
\bQ\left(\dtwovr(V,R),W,R\right)
-
\bQ\left(\donev(V),W,R\right)
\right\}
\right].
\]
Thus, the population-level value contrast is entirely driven by individuals
for whom the two rules assign different treatments. Its magnitude reflects
both the prevalence of such rule discordance and the conditional outcome
contrast between the treatments assigned by the two rules among individuals
for whom they disagree.

\subsection{Estimation}
We follow a targeted learning approach to estimate pragmatic effectiveness
as a contrast between the values of treatment rules, building on Targeted
Maximum Likelihood Estimation (TMLE) and cross-validated TMLE (CV-TMLE) for learning the mean outcome under the optimal treatment rules
\citep{van2006targeted,van2011targeted,van2015targeted,luedtke2016super,van2018targeted}.
The treatment rules under comparison may be specified \textit{a priori} by the user, or learned from the observed data possibly with sample-splitting procedures. When the rules are prespecified
or otherwise treated as fixed, we use TMLE to estimate their pragmatic
effectiveness. When the rules are learned from the same observed data,
particularly for the optimal-rule contrast of primary interest, we use
CV-TMLE to separate rule learning from rule evaluation and mitigate
finite-sample bias from reusing the same data for both tasks.

Because estimation of our primary optimal-rule estimand requires learning the optimal treatment rules with and without the medical test results from data, we focus our theoretical development below on the CV-TMLE procedure,
while presenting the TMLE first to provide methodological context in the simpler
setting with fixed treatment rules of interest.

\paragraph{Efficient influence curve.}
Let $g:\cA\times\cW\times\cR\to[0,1]$ 
denote the treatment mechanism, defined by
$
g(a\mid w,r):=P(A=a\mid W=w,R=r)$. 
Given $\donev$ and $\dtwovr$, the efficient influence curve (EIC) of the pragmatic effectiveness estimand $\psi^{\mathrm{prag}}_{\donev,\dtwovr}(Q)$ is
\begin{eqnarray}
D^{\mathrm{prag}}_{\donev,\dtwovr}(Q,g) &=& \frac{I\left(A=\dtwovr(V, R)\right)-I\left(A=\donev(V)\right)}{g(A \mid W,R)}\left[Y-\bQ(A,W,R)\right] \nonumber \\
&&+ \, \bQ(\dtwovr(V,R),W,R)-\bQ(\donev(V),W,R) - \psi^{\mathrm{prag}}_{\donev,\dtwovr}(Q).
\end{eqnarray}

\paragraph{Targeted Maximum Likelihood Estimation with given rules.}
For a probability distribution $P$ and measurable function $f$, we use
the shorthand
$
Pf:=\int f(o)\,dP(o),
$
and let $P_n$ denote the empirical distribution of the observed sample. 
Let $P_0\in\cM$ denote the
true observed-data distribution. Write
$
Q_0:=Q(P_0),
$
with the corresponding outcome regression 
$
\bQ_0(a,w,r)
:=
E_{P_0}\!\left[Y\mid A=a,W=w,R=r\right],
$ and let
$
g_0(a\mid w,r)
:=
P_0(A=a\mid W=w,R=r).
$
We first consider the target estimand as the pragmatic effectiveness indexed by the given rules 
\(\psi^{\mathrm{prag}}_{\donev,\dtwovr}(Q_0)\). 
We introduce the TMLE procedure following \cite{van2015targeted}.

First, we construct an initial estimator $\bQ_n(a,w,r)$ of the true outcome regression 
$
\bQ_0(a,w,r)
=
E_{P_0}\!\left[Y\mid A=a,W=w,R=r\right],
$
and an estimator $g_n(a\mid w,r)$ of
$
g_0(a\mid w,r).
$

Second, we define the clever covariate
\[
H_n(A,W,R)
:=
\frac{
I\left(A=\dtwovr(V,R))-I(A=\donev(V)\right)
}{
g_n\left(A\mid W,R\right)
}
\]
and consider the logistic fluctuation
submodel for updating the initial fit $\bQ_n$:
\begin{align}
\logit\left\{
\bQ_{n,\epsilon}(a,w,r)
\right\}
&=
\logit\left\{
\bQ_n(a,w,r)
\right\}
+
\epsilon H_n(a,w,r),
\label{eq:tmle_fluctuation_pragmatic}
\end{align}
The fluctuation parameter $\epsilon$ is then estimated by fitting a logistic regression
of \(Y\) on \(H_n(A,W,R)\), with
\(\logit\{\bQ_n(A,W,R)\}\) included as an offset. This targeting step solves the score
equation
$
P_n
\left[
H_n(A,W,R)
\left\{
Y-\bQ_n^*(A,W,R)
\right\}
\right] =
0$.
The updated outcome regression is given by
$\bQ_n^*(a,w,r)
:=
\bQ_{n,\epsilon_n}(a,w,r)$.

Third, let $Q_{n,W,R}$ denote the empirical distribution of $(W,R)$,
and let $Q_{n,Y}^*$ denote the updated conditional distribution of $Y$
corresponding to $\bQ_n^*$. Define
$
Q_n^*:=(Q_{n,W,R},Q_{n,Y}^*).
$
The TMLE of pragmatic effectiveness is then given by the plug-in estimate
\begin{align}
\psi^{\mathrm{prag}}_{\donev,\dtwovr}(Q_n^*)
&=
P_n
\left[
\bQ_n^*
\left\{
\dtwovr(V,R),W,R
\right\}
-
\bQ_n^*
\left\{
\donev(V),W,R
\right\}
\right].
\label{eq:tmle_pragmatic}
\end{align}

Under a strong positivity condition and standard regularity conditions \citep{van2015targeted},
the TMLE is asymptotically normal and efficient.
The standard error of TMLE is given by
$\sigma_n/\sqrt{n}$, where 
$
\sigma_n^2
=
P_n
\left[
D^{\mathrm{prag}}_{\donev,\dtwovr}(Q_n^*,g_n)
^2
\right]$ with $D^{\mathrm{prag}}_{\donev,\dtwovr}(Q_n^*,g_n)$ the estimated EIC given by
\begin{align}
D^{\mathrm{prag}}_{\donev,\dtwovr}(Q_n^*,g_n)
&:=
H_n(A,W,R)
\left[
Y-\bQ_n^*(A,W,R)
\right]
\nonumber\\
&\quad+
\bQ_n^*
\left(
\dtwovr(V,R),W,R
\right)
-
\bQ_n^*
\left(
\donev(V),W,R
\right)
-
\psi^{\mathrm{prag}}_{\donev,\dtwovr}(Q_n^*).
\label{eq:estimated_eic_pragmatic}
\end{align}
The corresponding \(100(1-\alpha)\%\) Wald-type confidence interval is
$\psi^{\mathrm{prag}}_{\donev,\dtwovr}(Q_n^*)
\pm
z_{1-\alpha/2}
\sigma_n/\sqrt{n}$.

The preceding TMLE procedure applies when the treatment rules are prespecified or otherwise treated as fixed. We next consider data-adaptively learned treatment rules and the corresponding CV-TMLE procedure.

\paragraph{Cross-Validated TMLE with (optimal) treatment rules learned adaptively from data.}
Consider data-adaptively learned treatment rules
$(\donevn,\dtwovrn)$. Conditional on the learned rules, the corresponding pragmatic
effectiveness estimand is 
$
\psi^{\mathrm{prag}}_{\donevn,\dtwovrn}(Q_0)$. 
This also applies to our
pragmatic effectiveness of interest with respect to optimal treatment rules, i.e., \(\psi^{\mathrm{prag}}_{\doptvn,\doptvrn}(Q_0)\), where $\doptvn$ and $\doptvrn$ are estimated optimal personalized treatment rules corresponding to the oracle rules  $\doptv$ and $\doptvr$, respectively.
However, since using the same data to both learn the optimal rules and estimate the expected outcome under the optimal treatment rule may induce upward bias in finite samples, we
introduce the general cross-validated TMLE (CV-TMLE) approach \citep{zheng2011cross,van2015targeted} to separate rule learning from rule
evaluation to mitigate such bias for the pragmatic effectiveness with respect to optimal treatment rules.

Let
$
\Val_1,\ldots,\Val_J$
denote a partition of the observed sample into \(J\) mutually exclusive
validation folds of equal size. For each fold \(j\), let
$
\Train_j
:=
\{1,\ldots,n\}\setminus \Val_j$
denote the corresponding training sample,
and let \(P_{n,j}\) denote the empirical distribution over observations
in the validation set \(\Val_j\).
For each fold \(j\), the optimal treatment rules (OTR) are learned using only
the training sample \(\Train_j\). To learn the OTRs, we first estimate conditional average treatment effect (CATE) functions under the reference and enriched information sets. 
These CATE functions are identified from the observed data under the exchangeability and
positivity assumptions above and can be estimated using doubly robust
pseudo-outcomes based on estimators of the outcome
regression and treatment mechanism
\citep{van2006statistical,luedtke2016super}. Specifically, we first construct
$
\bQ_{n,j}
$
to estimate
$
\bQ_0(a,w,r)
=
E_{P_0}(Y\mid A=a,W=w,R=r),
$
and construct
$
g_{n,j}
$
to estimate
$
g_0(a\mid w,r)
=
P_0(A=a\mid W=w,R=r)
$
using the training sample $\Train_j$. Then, we define the doubly robust
pseudo-outcome
\[
\xi_{n,j}(O)
:=
\bQ_{n,j}(1,W,R)-\bQ_{n,j}(0,W,R)
\nonumber + 
\frac{2A-1}{g_{n,j}(A\mid W,R)}
\left[
Y-\bQ_{n,j}(A,W,R)
\right],
\]
whose conditional expectation given $(W,R)$ recovers the CATE
conditional on $(W,R)$ if either \(\bQ_{n,j}=\bQ_0\) or \(g_{n,j}=g_0\).

We then estimate the enriched-information CATE  $\tauvr$ by regressing
\(\xi_{n,j}(O)\) on \((V,R)\). 
Similarly, the baseline-information CATE is estimated by regressing the
same pseudo-outcome on \(V\) alone.
The corresponding fold-specific estimated optimal treatment rules are
\begin{align}
\dtwovrnminusj(v,r)
&:=
I\left(
\tauvr_{n,j}(v,r)>0
\right),
\label{eq:estimated_opt_rule_wr_cate}
\\
\donevnminusj(v)
&:=
I\left(
\tauv_{n,j}(v)>0
\right).
\label{eq:estimated_opt_rule_w_cate}
\end{align}

The target estimand is the average of the fold-specific, data-adaptive
pragmatic effectiveness estimands:
\begin{align}
\tilde{\psi}_{n}^{\mathrm{prag}}(Q)
&:=
\frac{1}{J}
\sum_{j=1}^{J}
\psi^{\mathrm{prag}}_{\donevnminusj,\dtwovrnminusj}(Q).
\label{eq:cvtmle_data_adaptive_target}
\end{align}
Here, \(\donevnminusj\) and \(\dtwovrnminusj\) are learned using only the
corresponding training sample \(\Train_j\), and
\(\psi^{\mathrm{prag}}_{\donevnminusj,\dtwovrnminusj}(Q)\) evaluates the
population-level value contrast for that pair of learned rules. 
The
subscript \(n\) emphasizes that 
\(\tilde{\psi}_{n}^{\mathrm{prag}}(Q)\) is a data-adaptive estimand, depending on the
observed data through the fold-specific rules generated by the sample-splitting procedure.

At the true observed-data distribution, the corresponding
cross-validated data-adaptive target is
$
\tilde{\psi}_{n}^{\mathrm{prag}}(Q_0).
$
We implement the CV-TMLE approach to estimate $\tilde{\psi}_{n}^{\mathrm{prag}}(Q_0)$ as follows.
First, for fold \(j\), define the fold-specific clever covariate
\begin{align}
H_{n,j}(A,W,R)
&:=
\frac{
I\left(A=\dtwovrnminusj(V,R)\right) - I\left(A=\donevnminusj(V)\right)
}{
g_{n,j}\left(A \mid W,R\right)
}
.
\label{eq:cvtmle_clever_covariate_pragmatic}
\end{align}

We then define the fold-specific logistic fluctuation submodel
\begin{align}
\logit\left\{
\bQ_{n,j,\epsilon}(a,w,r)
\right\}
&=
\logit\left\{
\bQ_{n,j}(a,w,r)
\right\}
+
\epsilon H_{n,j}(a,w,r).
\label{eq:cvtmle_fluctuation_pragmatic}
\end{align}
While the fluctuation submodel is fold-specific through
$H_{n,j}$ and $\bQ_{n,j}$, a single fluctuation parameter is estimated
jointly across all folds. Let $j(i)$ denote the validation fold containing
observation $i$. We stack the validation observations across folds and fit
one pooled logistic regression of $Y_i$ on
$H_{n,j(i)}(A_i,W_i,R_i)$, using
$\logit\left\{\bQ_{n,j(i)}(A_i,W_i,R_i)\right\}$ as an offset.
Denote the resulting estimate by $\epsilon_{n,\mathrm{cv}}$.
For each fold $j$, the updated outcome regression is
$
\bQ_{n,j}^{*}(a,w,r)
:=
\bQ_{n,j,\epsilon_{n,\mathrm{cv}}}(a,w,r).$
These updates solve 
$
\frac{1}{J}
\sum_{j=1}^{J}
P_{n,j}
\left[
H_{n,j}(A,W,R)
\left\{
Y-\bQ_{n,j}^{*}(A,W,R)
\right\}
\right]
=
0.$

Define \(Q_{n,j,W,R}\) as the empirical
distribution of \((W,R)\) in each validation fold \(\Val_j\), and
let \(Q_{n,j,Y}^{*}\) denote the updated conditional distribution of
\(Y\) corresponding to \(\bQ_{n,j}^{*}\). Write
$
Q_{n,j}^{*}:=(Q_{n,j,W,R},Q_{n,j,Y}^{*})$.
The fold-specific plug-in estimate of the pragmatic effectiveness is
\begin{align}
\psi^{\mathrm{prag}}_{\donevnminusj,\dtwovrnminusj}(Q_{n,j}^{*})
&=
P_{n,j}
\left[
\bQ_{n,j}^{*}
\left(
\dtwovrnminusj(V,R),W,R
\right)
-
\bQ_{n,j}^{*}
\left(
\donevnminusj(V),W,R
\right)
\right].
\label{eq:cvtmle_fold_pragmatic}
\end{align}

The CV-TMLE is the average
of these fold-specific plug-in estimates,
\begin{align}
\psi_{n,\mathrm{cv-tmle}}^{\mathrm{prag},*}
&:=
\frac{1}{J}
\sum_{j=1}^{J}
\psi^{\mathrm{prag}}_{\donevnminusj,\dtwovrnminusj}(Q_{n,j}^{*}).
\label{eq:cvtmle_pragmatic}
\end{align}

For each validation fold $j$, define the estimated fold-specific efficient
influence curve
$
D_{n,j}^{\mathrm{prag},*}(O)
:=
D_{\donevnminusj,\dtwovrnminusj}^{\mathrm{prag}}
\left(
Q_{n,j}^{*},g_{n,j}
\right)(O)$.
The standard error estimate is 
$\sigma_{n,\mathrm{cv-tmle}}/\sqrt{n}$, where 
$
\sigma_{n,\mathrm{cv-tmle}}^2
:=
\frac{1}{J}
\sum_{j=1}^{J}
P_{n,j}
\left[
D_{n,j}^{\mathrm{prag},*}(O)^2
\right]$.
The \(100(1-\alpha)\%\) Wald-type confidence interval is given by 
$
\psi_{n,\mathrm{cv-tmle}}^{\mathrm{prag},*}
\pm
z_{1-\alpha/2}
\sigma_{n,\mathrm{cv-tmle}}/\sqrt{n}$.

\paragraph{Asymptotic properties of the CV-TMLE.}
We next state sufficient conditions for asymptotic inference
based on the proposed CV-TMLE, following Theorem 6 of \cite{van2015targeted}. A general CV-TMLE representation specialized
to the pragmatic effectiveness estimand, together with verification of the
conditions below, is provided in the Appendix.

\begin{assumption}[Strong positivity]
\label{assump:cvtmle_positivity}
There exists a constant $\delta>0$ such that
\[
P_0
\left[
\inf_{a\in\cA}
g_0
\left(
a\mid W,R
\right)
\geq \delta
\right]
=
1,
\]
and, with probability tending to one,
\[
\inf_{j\in\{1,\ldots,J\}}
\inf_{a\in\cA}
g_{n,j}
\left(
a\mid W,R
\right)
\geq \delta
\qquad
P_0\text{-a.s.}
\]
\end{assumption}

\begin{assumption}[Convergence of the cross-validated efficient influence curve]
\label{assump:cvtmle_eic_convergence}
There exist fixed limiting treatment rules
$d_{1,\infty}^{v}:\cV\to\cA$ and
$d_{2,\infty}^{v,r}:\cV\times\cR\to\cA$,
a possibly misspecified limiting targeted distribution $Q_{\infty}$,
and a limiting treatment mechanism $g_{\infty}$ such that
\[
\max_{j\in\{1,\ldots,J\}}
P_0
\left[
\left\{
D_{n,j}^{\mathrm{prag},*}
-
D_{\infty}^{\mathrm{prag}}
\right\}^{2}
\right]
=
o_p(1),
\]
where
$
D_{\infty}^{\mathrm{prag}}(O)
:=
D_{d_{1,\infty}^{v},d_{2,\infty}^{v,r}}^{\mathrm{prag}}
\left(
Q_{\infty},g_{\infty}
\right)(O)
$.
\end{assumption}

\begin{assumption}[Second-order convergence]
\label{assump:cvtmle_second_order}
The nuisance estimators satisfy
\[
\max_{j\in\{1,\ldots,J\}}
\left\|
\bQ_{n,j}^{*}-\bQ_0
\right\|_{P_0,2}
\left\|
g_{n,j}-g_0
\right\|_{P_0,2}
=
o_p
\left(
n^{-1/2}
\right).
\]
\end{assumption}

\begin{theorem}[Asymptotic normality of the pragmatic-effectiveness CV-TMLE]
\label{thm:cvtmle_pragmatic_an}
Under Assumptions~\ref{assump:cvtmle_positivity}, \ref{assump:cvtmle_eic_convergence} and \ref{assump:cvtmle_second_order},
\[
\psi_{n,\mathrm{cv-tmle}}^{\mathrm{prag},*}
-
\tilde{\psi}_{n}^{\mathrm{prag}}(Q_0)
=
(P_n-P_0)
D_{\infty}^{\mathrm{prag}}
+
o_p
\left(
n^{-1/2}
\right).
\]
Consequently,
\[
\sqrt{n}
\left[
\psi_{n,\mathrm{cv-tmle}}^{\mathrm{prag},*}
-
\tilde{\psi}_{n}^{\mathrm{prag}}(Q_0)
\right]
\dto
N
\left(
0,\sigma_{\infty}^{2}
\right),
\]
where
$
\sigma_{\infty}^{2}
:=
\Var_{P_0}
\left\{
D_{\infty}^{\mathrm{prag}}(O)
\right\}$.

If, in addition,
$
\tilde{\psi}_{n}^{\mathrm{prag}}(Q_0)
-
\psi_{\doptv,\doptvr}^{\mathrm{prag}}(Q_0)
=
o_p
\left(
n^{-1/2}
\right)$,
then 
\[
\sqrt{n}
\left[
\psi_{n,\mathrm{cv-tmle}}^{\mathrm{prag},*}
-
\psi_{\doptv,\doptvr}^{\mathrm{prag}}(Q_0)
\right]
\dto
N
\left(
0,\sigma_{\infty}^{2}
\right).
\]
\end{theorem}

\section{Simulation study}
We conduct a simulation study to evaluate the finite-sample performance
of the TMLE-based approaches for assessing the explanatory effectiveness and the pragmatic effectiveness under
learned optimal treatment rules. 

We consider two baseline covariates
$W=(W_1,W_2)$,
where
$
W_1\sim \mathrm{Bernoulli}(0.5)$,
$W_2\sim \mathrm{Uniform}(-3,3)$.
The diagnostic test result \(R\in\{-1,1\}\) was generated by
$
P(R=1\mid W)
=
\min\left\{
\max\left\{
1/[1+\exp(-(0.4W_1+0.3W_2))],
\,0.2
\right\},
\,0.8
\right\}.
$
Treatment was assigned according to
$
P(A=1\mid W,R) = \min\left\{
\max\left\{1/[1+\mathrm{exp}(-(0.35W_1+0.55R))],0.2
\right\},
\,0.8
\right\}.
$ 
The conditional mean outcome was
$
E(Y\mid W,R,A) = W_1+1.5 W_2 \times A+2.5R \times A$.
The outcome $Y$ is generated by adding Gaussian noise ($\sim N(0,1)$) and 
is bounded between -10 and 10. 
We consider four sample sizes, $n\in\{250,500,1000,2000\}$,
and run 500 Monte Carlo simulations for each sample size. The TMLE approach is implemented by rescaling the outcomes to lie between 0 and 1 using the sample minimum and maximum, and then transforming the estimates back to the original scale, following \cite{gruber2010targeted}. 

We first report the explanatory effectiveness of \(R\) under this data-generating distribution using the variable importance measure (VIM) described above, where we set $V=(W_1,W_2)$. Based on the simulation setup, the true VIM-based explanatory effectiveness of \(R\) is 5.77, while the corresponding VIMs for the baseline covariates \(W_1\) and \(W_2\) are 0 and 6.3, respectively. These values indicate that \(R\) has substantial explanatory effectiveness in modifying the conditional average treatment effect. We also implement TMLE of VIM following the estimation procedures described by \citet{li2023targeted}; the estimator exhibits low bias, decreasing variance with increasing sample size, and empirical coverage close to the nominal level (Table~\ref{tab:simulation_exp}). 
\begin{table}[H]
\centering
\caption{Simulation results of TMLE for explanatory effectiveness of $R$ using VIM}
\label{tab:simulation_exp}
\begin{tabular}{cccccc}
\toprule
\(n\) &
Truth &
Bias &
Variance &
Coverage \\
\midrule
250  & 5.77 & 0.092 & 0.476 & 0.940 \\
500  & 5.77 & 0.100 & 0.215 & 0.948 \\
1000 & 5.77 & 0.052 & 0.114 & 0.940 \\
2000 & 5.77 & 0.038 & 0.049 & 0.950 \\
\bottomrule
\end{tabular}
\end{table}

We then evaluate our proposed estimand of pragmatic effectiveness by assessing whether incorporating the diagnostic test result $R$ into the information set $V=(W_1,W_2)$ to estimate the optimal personalized treatment rule $\doptvr$ improves the expected outcomes, compared with the optimal personalized treatment rule $\doptv$ using baseline covariates $W_1$ and $W_2$ only.
Specifically, for each simulation dataset, we use sample splitting to define the data-adaptive target estimand and implement CV-TMLE to obtain point estimates and conduct statistical inference with 95\% Wald-type confidence intervals. For estimation of outcome regression and CATE
functions, we use Super Learner \citep{van2007super} with a library consisting of a linear model without interaction terms, an intercept-only
model, and random forest \citep{breiman2001random}. For estimation of the treatment assignment mechanism, we use
Super Learner with a library of a logistic
regression without interaction terms, an intercept-only model and random forest. 

Table~\ref{tab:simulation_prag} summarizes the performance of the
CV-TMLE for pragmatic effectiveness.
The average value of the target estimand across simulations is approximately 0.314 for all sample sizes. The corresponding mean outcomes under the learned reference-information and enriched-information optimal treatment rules are approximately 1.71--1.77 and 2.02--2.08, respectively, while the average discordance rate between the estimated rules is around 25\%.
The bias relative to the true estimand decreases toward zero as the sample size increases. The variance also decreases with increasing sample size. Empirical coverage
shows finite-sample undercoverage at the smallest sample size \(n=250\) and approaches the
nominal 95\% level for $n \in \{500,1000,2000\}$.  

\begin{table}[H]
\centering
\caption{Simulation results of CV-TMLE for pragmatic effectiveness of $R$ in addition to $W_1$ and $W_2$. Target, Est., Bias, Var., Cov., and Disc. denote the Monte Carlo mean of the pragmatic estimand, mean estimate, mean bias, empirical variance, empirical coverage, and mean discordance rate between estimated  \(d_{\mathrm{opt}}^{v,r}\) and \(d_{\mathrm{opt}}^v\), respectively. The last two columns report the mean estimated values of the two learned rules.}
\label{tab:simulation_prag}
\vspace{1em}
\begin{tabular}{ccccccccc}
\toprule
\(n\) & Target & Mean Est. & Bias & Var. & Cov. & Disc. & \(\theta(\doptv)\) & \(\theta(\doptvr)\) \\
\midrule
250  & 0.314 & 0.310 & -0.0042 & 0.0072 & 0.914 & 0.251 & 1.713 & 2.022\\
500  & 0.314 & 0.316 & 0.0021 & 0.0036 & 0.944 & 0.248 & 1.744 & 2.061\\
1000 & 0.314 & 0.315 & 0.0005 & 0.0018 & 0.948 & 0.244 & 1.762 & 2.077\\
2000 & 0.314 & 0.314 & -0.0001 & 0.0007 & 0.964 & 0.241 & 1.767 & 2.081\\
\bottomrule
\end{tabular}

\end{table}

\section{Data application}
We applied the explanatory and pragmatic effectiveness estimators to a synthetic colorectal cancer dataset that mimics the context of the PRIME randomized phase III trial, comparing the effectiveness of panitumumab plus FOLFOX4 versus FOLFOX4 alone for first-line treatment of metastatic colorectal cancer \citep{douillard2014final}. In PRIME, patients were randomized to panitumumab plus FOLFOX4 or FOLFOX4 alone, and KRAS mutation status (wild-type or mutant) was evaluated as a predictive biomarker for anti-EGFR therapy. 

Our analysis applies the proposed methods to a synthetic dataset constructed to resemble the PRIME setting, which may be interpreted as an illustrative PRIME-like application rather than as an analysis or reproduction of the original trial results. We coded the treatment as \(A=1\) for panitumumab plus FOLFOX4 and \(A=0\) for FOLFOX4 alone. We used KRAS as the diagnostic test result of interest, encoding \(R=0\) for KRAS wild-type and \(R=1\) for KRAS mutant. We restricted the primary analyses to uncensored observations ($n=$ 603) and treated the observed survival (OS) time as a continuous outcome for this illustrative example, with larger values indicating better outcomes. 

We considered two reference information sets $V$ to estimate $\doptv$ and $\doptvr$. The first set $V_1 = \emptyset$ uses no baseline covariates in the rule class, so that the enriched rule $\doptvr$ may depend only on KRAS status $R$. This setting evaluates the value of biomarker-guided treatment decisions relative to a single marginal treatment recommendation. The second set, \(V_2\), includes age, region, and an Eastern Cooperative Oncology Group (ECOG) performance status with possible discrete values in $\{0,1,2\}$, where region and ECOG correspond to randomization stratification factors and age is an additional clinically relevant baseline characteristic in the PRIME study.
Separately, nuisance functions \(E(Y\mid A,W,R)\) and \(P(A=1\mid W,R)\) were fitted using the richer available baseline covariate set $W$, consisting of region, age, sex, race, metastatic disease indicators, liver-only metastasis, lactate dehydrogenase, ECOG performance status, and primary tumor location, with the same Super
Learner library as in the simulation study.

\paragraph{Explanatory effectiveness.}
We first report the explanatory effectiveness of KRAS status based on VIM (Table~\ref{tab:application_vim}), using $V_1$ and $V_2$ as baseline covariate subsets of interest, respectively.
When no baseline covariates were included ($V_1 =\emptyset$), KRAS status shows positive VIM for OS. Using the covariate set $V_2$, the estimated explanatory effectiveness of KRAS status was substantially larger, which suggests that KRAS status carries treatment-relevant information especially after conditioning on other baseline covariates in this synthetic dataset.

\begin{table}[ht]
\centering
\caption{Explanatory effectiveness of KRAS status $R$ in the synthetic PRIME study with different baseline covariates $V_1$ and $V_2$.}
\label{tab:application_vim}
\vspace{1em}
\begin{tabular}{ccccc}
\toprule
Baseline Covariate Set & Estimate & SE & 95\% CI \\
\midrule
 \(V_1 = \emptyset\)  & 2673.59 & 1096.64 & (524.18, 4823.00) \\
 \(V_2 = \{\mathrm{region, age, ECOG}\}\)  &  14671.14 & 2138.11 & (10480.45, 18861.83) \\
\bottomrule
\end{tabular}
\end{table}

\paragraph{Pragmatic effectiveness.}
We next estimated pragmatic effectiveness using CV-TMLE. Table~\ref{tab:application_prag} reports the estimated pragmatic effectiveness in two scenarios that use $V_1$ and $V_2$, to construct the baseline OTR $\doptv$, respectively.  

In the setting with \(V_1 = \emptyset\), the estimated pragmatic effectiveness for OS was 12.61 days, with the KRAS-informed optimal treatment rule changing treatment recommendations for approximately \(52\%\) of observations. When the reference information set was enlarged to \(V_2\), the estimated pragmatic effectiveness decreased to 8.21 days, and the discordance rate decreased to approximately \(20.9\%\). This reduction reflects the smaller incremental role of KRAS status when a richer reference information set is already available for treatment personalization, leaving less room for KRAS status to further alter optimal treatment recommendations. While both point estimates were positive, their 95\% confidence intervals
included zero and were therefore not statistically distinguishable from zero in this setting.

\begin{table}[ht]
\centering
\caption{Estimate of pragmatic effectiveness $\psi_{n,\mathrm{cv-tmle}}^{\mathrm{prag},*}$ of KRAS-guided treatment rules in the synthetic PRIME study application.}
\resizebox{\textwidth}{!}{%
\begin{tabular}{cccccccc}
\toprule
$V$ &  Estimate & SE & 95\% CI & Disc. & \(\theta(\doptv)\) & \(\theta(\doptvr)\) \\
\midrule
$V_1 = \emptyset$  &  12.61 & 10.56 & (-8.08, 33.31) & 0.522 & 423.58 & 436.19 \\
$V_2 = \{\mathrm{region, age, ECOG}\}$ & 8.21 & 6.55 & (-4.63, 21.06) & 0.209 & 432.81 & 441.02 \\
\bottomrule
\end{tabular}%
}
\label{tab:application_prag}
\end{table}

Overall, the application illustrates the distinction between explanatory and pragmatic effectiveness. KRAS status showed substantial explanatory effectiveness, indicating meaningful treatment-effect heterogeneity associated with the biomarker. However, the estimated pragmatic effectiveness was more modest. This reflects the decision-theoretic interpretation of the proposed estimand: even when a diagnostic test substantially explains treatment-effect heterogeneity, its pragmatic value depends on whether the additional information changes treatment recommendations in a way that improves expected outcomes. The pragmatic estimand therefore complements explanatory effectiveness by providing decision-relevant evidence on the actual outcome benefit gained from using the diagnostic test to support individualized treatment decision-making and optimization.

\section{Application to Evaluation of AI-Enabled Devices and Real-World Evidence}
\label{section:discussion}

The proposed methodology provides a principled causal lens for evaluating AI-enabled devices, including both diagnostic information tools and action-oriented decision-support tools. First, an AI-enabled device may generate an additional diagnostic, prognostic, or risk
signal that enriches the information available for treatment decisions. Second,
it may directly assist with, recommend, or optimize the treatment decision rule
itself. To make this distinction explicit, we introduce a two-dimensional approach that separates information augmentation from treatment-rule optimization as two pathways for supporting clinical decision-making, and then use this approach to inform the evaluation of AI-enabled devices.

\subsection{Distinguishing information augmentation and treatment-rule optimization}

The main pragmatic effectiveness estimand in this paper focuses on the value of
diagnostic information under optimized personalized treatment rules, namely the
contrast between \(\theta(\doptvr)\) and \(\theta(\doptv)\). This contrast isolates
the gain from enriching the information set from \(V\) to \((V,R)\) when the
rule-generating principle is held fixed as outcome-optimal treatment
personalization. More generally, the same rule-comparison approach can be used to evaluate other clinically relevant treatment-rule contrasts.
In practice, treatment decisions are often governed by current clinical
practice, guideline-based recommendations, physician judgment, or generated by some baseline
algorithms that are not necessarily fully optimized for the outcome of
interest. We refer to such benchmark treatment decision mechanisms as \emph{reference
rules}. Let \(\drefv:\cV\to\cA\) and
\(\drefvr:\cV\times\cR\to\cA\) denote reference treatment rules under the
reference information \(V\) and the enriched information \((V,R)\), respectively.
These rules play the role of clinically meaningful comparators, whereas
\(\doptv\) and \(\doptvr\) denote the corresponding optimized rules that maximize
expected outcomes under the two information sets. Together, these four rules
induce the following \(2\times 2\) information-by-decision value structure:

\begin{table}[H]
\centering
\caption{Two-dimensional approach for information augmentation and treatment-rule optimization.}
\label{tab:decision_value_approach}
\[
\begin{array}{c|cc}
& \text{Reference information } V
& \text{Enriched information } (V,R) \\
\hline
\text{Reference rules}
& \theta(\drefv)
& \theta(\drefvr) \\
\text{Optimized rules}
& \theta(\doptv)
& \theta(\doptvr)
\end{array}
\]
\end{table}

This approach organizes potential improvements in outcomes along two conceptual
dimensions. The first dimension is information augmentation: the gain obtained
when treatment decisions are allowed to use enriched information \((V,R)\)
instead of reference information \(V\) alone. The second dimension is
treatment-rule optimization: the gain obtained by replacing a reference rule with
an optimized rule while holding the information set fixed.

For example, the total improvement obtained by moving from baseline reference
practice \(\drefv\) to optimized decision-making with enriched information
\(\doptvr\) can be decomposed as
\[
\theta(\doptvr) - \theta(\drefv)
=
\underbrace{
\big[\theta(\doptv) - \theta(\drefv)\big]
}_{
\substack{
\text{treatment-rule optimization}\\
\text{using reference information}
}
}
+
\underbrace{
\big[\theta(\doptvr) - \theta(\doptv)\big]
}_{
\substack{
\text{information augmentation}\\
\text{under optimized rules}
}
},
\]
or equivalently,
\[
\theta(\doptvr) - \theta(\drefv)
=
\underbrace{
\big[\theta(\drefvr) - \theta(\drefv)\big]
}_{
\substack{
\text{information augmentation}\\
\text{under reference rules}
}
}
+
\underbrace{
\big[\theta(\doptvr) - \theta(\drefvr)\big]
}_{
\substack{
\text{treatment-rule optimization}\\
\text{using enriched information}
}
}.
\]

These decompositions support interpretation of the two pathways through which diagnostic information influences expected outcomes. The
first path asks how much is gained by first optimizing the treatment rule using
reference information and then adding the diagnostic test result to an already
optimized decision rule. The second path asks how much is gained by first
allowing the reference rule to use the diagnostic test result and then optimizing
the rule under the enriched information set. 

While the structure above is presented in terms of deterministic treatment
rules, the same logic also applies to stochastic decision mechanisms. For
example, under usual care, patients with the same reference information \(V\)
may receive different treatments because of variation across physicians,
institutions, patient preferences, or local practice patterns. Similarly, after
the diagnostic test result is observed, patients with the same enriched
information \((V,R)\) may still receive different treatments because clinicians
may interpret or act on the information differently. Such mechanisms can be represented by stochastic policies
$\pi^v:\mathcal A\times\mathcal V\to [0,1]$, and $\pi^{v,r}:\mathcal A\times\mathcal V\times\mathcal R\to [0,1],$
where \(\pi^v(a\mid v)\) and \(\pi^{v,r}(a\mid v,r)\) denote the probability
of assigning treatment \(a\) conditional on the available information \(v\) or
\((v,r)\), respectively. 
The value-comparison formulation therefore extends naturally from deterministic
rules to stochastic decision mechanisms. The TMLE and CV-TMLE procedures can be
adapted to such stochastic-policy contrasts by replacing the deterministic rule
indicators in the target estimand and clever covariate with the corresponding
stochastic policy probabilities (see, for example, \citet{munoz2012population,diaz2018stochastic}).

\subsection{Evaluating Dual Roles of AI-Enabled Devices as Diagnostic and Decision-Optimization Tools}

This information-augmentation-by-treatment-optimization approach provides a natural language for
evaluating AI-enabled clinical decision support. Information-oriented AI tools,
such as diagnostic, prognostic, biomarker, imaging, or risk models, can be
viewed as generating an enriched signal \(R_{\mathrm{AI}}\) that augments the
information available for treatment decisions. In this case, the explanatory
estimand asks whether \(R_{\mathrm{AI}}\) captures treatment-effect
heterogeneity beyond reference information, while the pragmatic estimand asks
whether incorporating this signal into treatment decisions improves expected
outcomes. 
Action-oriented AI tools, such as systems supporting treatment selection, dosing, or personalization, can instead be viewed as modifying the rule used to translate available information into treatment decisions, for example from a reference rule \(d_{\mathrm{ref}}\) to an AI-supported rule \(d_{\mathrm{AI}}\). Their value can therefore be evaluated through the expected outcome under the AI-supported rule, relative to the reference rule. 

This distinction naturally extends to more comprehensive AI-enabled workflows in which both components are present: an AI-generated signal enriches the information available for clinical decision-making, and an AI-enabled rule uses that information to support or recommend treatment decisions. 
The value contrasts in Table~\ref{tab:decision_value_approach}, together with the corresponding value decompositions, can then be used to evaluate whether an AI-enabled device improves outcomes through the information it provides, the treatment decisions it supports, or their combination.

The perspective on the dual roles of AI-enabled devices can be considered alongside recent discussions of transparency for artificial-intelligence/machine-learning-enabled medical devices, which highlight the importance of understanding how such devices provide information and interact with clinical decision-making in health care workflows \citep{fda2024transparency,shick2024transparency}.
This causal inference approach may provide a quantitative basis for characterizing evaluation targets in the context of planned modifications as described in current guidance \citep{fda2025pccp}.
In particular, planned modifications that affect earlier detection or diagnosis, personalized diagnostics and therapeutics, or assistive decision-support functions can be interpreted according to whether they primarily change the information available for treatment decisions, the rule used to translate information into treatment recommendations, or both.

\section{Conclusions}
This paper develops a causal inference approach for evaluating diagnostic tests and AI-enabled clinical decision-support tools, with a particular focus on settings in which their effects on health outcomes operate through information and downstream treatment decisions. Unlike classical interventions such as drugs and therapeutic devices, diagnostic tests, biomarkers, and information-oriented AI/ML tools generally do not directly affect health outcomes. Their value depends on whether the information they provide changes downstream treatment decisions in ways that improve expected outcomes. This decision-mediated pathway makes diagnostic-test evaluation a distinct causal problem from the evaluation of treatments or interventions whose effects are assessed through direct contrasts in health outcomes. 

The proposed approach separates two questions that are often conflated: explanatory effectiveness and pragmatic effectiveness. Explanatory effectiveness evaluates whether a diagnostic signal provides
additional treatment-relevant information by explaining treatment-effect
heterogeneity beyond baseline covariates of interest. Pragmatic effectiveness evaluates whether incorporating diagnostic information in treatment decision-making improves expected outcomes. This distinction clarifies why evidence of effect modification alone does not establish clinical utility: a diagnostic signal may be predictive of treatment benefit, yet have limited value if it does not change treatment recommendations or if the resulting changes yield little outcome improvement.

We characterize the explanatory effectiveness using a variance-based treatment-effect variable importance measure, which summarizes the additional variation in conditional treatment effects explained by the diagnostic signal beyond baseline covariates of interest. We formalize the pragmatic effectiveness as a treatment-rule value contrast under different information sets and provide TMLE and CV-TMLE procedures for nonparametric estimation and inference. The simulation study supported the finite-sample performance of the proposed estimators, and the application to synthetic, PRIME-like clinical trial data illustrated the practical distinction between explanatory and pragmatic effectiveness. In particular, the application showed how a biomarker may carry substantial treatment-relevant information while producing a more modest gain when incorporated into optimized treatment decisions.

Our approach provides a unified causal lens for evaluating AI/ML-enabled clinical decision-support devices across broad application domains. Information-oriented tools, such as diagnostic or risk-monitoring models, can be evaluated by the treatment-effect heterogeneity they explain and by the downstream value of using that information in treatment decisions. This approach also naturally accommodates the evaluation of action-oriented tools, such as devices supporting treatment selection or personalization, by quantifying the improvement in expected outcomes achieved through the updated treatment rules embedded in or supported by such devices.

With a unified information-augmentation-by-treatment-optimization approach, this work provides a statistical foundation for formulating causal estimands that reflect different evaluation perspectives for diagnostic tests and AI-enabled devices in a nonparametric model, and for obtaining evidence from sources such as clinical trials and real-world data, without relying on parametric assumptions. This methodology may support application in domains such as health monitoring, clinical care delivery, personalized medicine, AI-deployment assessment, and other decision-making settings in which both enriched information and optimized decision rules contribute to improvements in target outcomes.

\bibliographystyle{plainnatbookauthor}
\bibliography{main}

\appendix
\section*{Appendix}
\section{Asymptotics of the CV-TMLE}
\label{appendix:cvtmle_theory}

We provide a specialization of the CV-TMLE representation and asymptotics in Theorem 6 of \citet{van2015targeted} to the pragmatic-effectiveness contrast considered
in this paper and verify that
Assumptions~\ref{assump:cvtmle_positivity}--%
\ref{assump:cvtmle_second_order}
imply the required conditions.

Recall that for each fold $j$, the fold-specific estimated EIC is 
$
D_{n,j}^{\mathrm{prag},*}
:=
D_{\donevnminusj,\dtwovrnminusj}^{\mathrm{prag}}
\left(
Q_{n,j}^{*},g_{n,j}
\right)$.
By construction of the pooled targeting step,
$
\frac{1}{J}
\sum_{j=1}^{J}
P_{n,j}
D_{n,j}^{\mathrm{prag},*}
=
0$.

For each fold $j$, define the fixed-rule second-order remainder
\begin{align}
\operatorname{Rem}_{n,j}^{\mathrm{prag}}
&:=
\psi_{\donevnminusj,\dtwovrnminusj}^{\mathrm{prag}}
\left(
Q_{n,j}^{*}
\right)
-
\psi_{\donevnminusj,\dtwovrnminusj}^{\mathrm{prag}}
\left(
Q_0
\right)
+
P_0
D_{n,j}^{\mathrm{prag},*}.
\label{eq:cvtmle_rem_appendix}
\end{align}

\begin{theorem}[Representation of CV-TMLE]
\label{thm:cvtmle_representation_appendix}
Suppose that,
$
\max_{j\in\{1,\ldots,J\}}
\left|
D_{n,j}^{\mathrm{prag},*}
\right| \leq M
$
for some $M<\infty$ with probability tending to one.
Suppose further that there exist fixed limiting treatment rules
$d_{1,\infty}^{v}$ and $d_{2,\infty}^{v,r}$,
a possibly misspecified limiting distribution $Q_{\infty}$,
and a limiting treatment mechanism $g_{\infty}$ such that
\[
\max_{j\in\{1,\ldots,J\}}
P_0
\left[
\left\{
D_{n,j}^{\mathrm{prag},*}
-
D_{\infty}^{\mathrm{prag}}
\right\}^{2}
\right]
=
o_p(1),
\]
where
$
D_{\infty}^{\mathrm{prag}}
:=
D_{d_{1,\infty}^{v},d_{2,\infty}^{v,r}}^{\mathrm{prag}}
\left(
Q_{\infty},g_{\infty}
\right)$.
Then,
\begin{align}
\psi_{n,\mathrm{cv-tmle}}^{\mathrm{prag},*}
-
\tilde{\psi}_{n}^{\mathrm{prag}}(Q_0)
=
(P_n-P_0)
D_{\infty}^{\mathrm{prag}}
+
\frac{1}{J}
\sum_{j=1}^{J}
\operatorname{Rem}_{n,j}^{\mathrm{prag}}
+
o_p
\left(
n^{-1/2}
\right).
\label{eq:cvtmle_representation_appendix}
\end{align}
\end{theorem}

\begin{proof}
For each fold $j$,
\[
\psi_{\donevnminusj,\dtwovrnminusj}^{\mathrm{prag}}
\left(
Q_{n,j}^{*}
\right)
-
\psi_{\donevnminusj,\dtwovrnminusj}^{\mathrm{prag}}
\left(
Q_0
\right)
=
-
P_0D_{n,j}^{\mathrm{prag},*}
+
\operatorname{Rem}_{n,j}^{\mathrm{prag}}.
\]
Averaging across folds and adding and subtracting the corresponding empirical means over the validation sets gives
\begin{align}
\psi_{n,\mathrm{cv-tmle}}^{\mathrm{prag},*}
-
\tilde{\psi}_{n}^{\mathrm{prag}}(Q_0)
&=
\frac{1}{J}
\sum_{j=1}^{J}
\left(
P_{n,j}-P_0
\right)
D_{n,j}^{\mathrm{prag},*}
-
\frac{1}{J}
\sum_{j=1}^{J}
P_{n,j}
D_{n,j}^{\mathrm{prag},*}
+
\frac{1}{J}
\sum_{j=1}^{J}
\operatorname{Rem}_{n,j}^{\mathrm{prag}}.
\end{align}
The second term on the right-hand side is zero due to the targeting
step that solves \linebreak $
\frac{1}{J}
\sum_{j=1}^{J}
P_{n,j}
D_{n,j}^{\mathrm{prag},*}
=
0$. Furthermore,
\begin{align}
&\frac{1}{J}
\sum_{j=1}^{J}
\left(
P_{n,j}-P_0
\right)
D_{n,j}^{\mathrm{prag},*}
\nonumber
=
(P_n-P_0)
D_{\infty}^{\mathrm{prag}}
+
\frac{1}{J}
\sum_{j=1}^{J}
\left(
P_{n,j}-P_0
\right)
\left(
D_{n,j}^{\mathrm{prag},*}
-
D_{\infty}^{\mathrm{prag}}
\right).
\end{align}
By the cross-validated empirical-process argument of
\citet{van2015targeted}, the uniform boundedness of
$D_{n,j}^{\mathrm{prag},*}$ together with its $L_2(P_0)$ convergence to
$D_{\infty}^{\mathrm{prag}}$ implies
\[
\frac{1}{J}
\sum_{j=1}^{J}
\left(
P_{n,j}-P_0
\right)
\left(
D_{n,j}^{\mathrm{prag},*}
-
D_{\infty}^{\mathrm{prag}}
\right)
=
o_p
\left(
n^{-1/2}
\right).
\]
Combining these results establishes
\eqref{eq:cvtmle_representation_appendix}.
\end{proof}

\begin{proof}[Proof of Theorem \ref{thm:cvtmle_pragmatic_an}]
We verify that
Assumptions~\ref{assump:cvtmle_positivity}--%
\ref{assump:cvtmle_second_order}
imply the conditions of
Theorem~\ref{thm:cvtmle_representation_appendix}
and that the remaining fixed-rule second-order remainder is
asymptotically negligible.

First, recall from the statistical setup that $Y\in[0,1]$.
The logistic fluctuation therefore preserves
$\bQ_{n,j}^{*}\in[0,1]$.
Under Assumption~\ref{assump:cvtmle_positivity},
$
\left|
\frac{
I\left(A=\dtwovrnminusj(V,R)\right)
-
I\left(A=\donevnminusj(V)\right)
}{
g_{n,j}(A\mid W,R)
}
\right|
\leq
\delta^{-1}
$
with probability tending to one. Since
$
\left|
Y-\bQ_{n,j}^{*}(A,W,R)
\right|
\leq 1$,
and both the conditional-mean contrast and its corresponding
fold-specific value are bounded in absolute value by one, it follows that
$
\max_{j\in\{1,\ldots,J\}}
\left\|
D_{n,j}^{\mathrm{prag},*}
\right\|_{\infty}
\leq
\delta^{-1}+2$
with probability tending to one. Thus, the uniform boundedness condition
of Theorem~\ref{thm:cvtmle_representation_appendix} is satisfied.

Second,
Assumption~\ref{assump:cvtmle_eic_convergence} is the
$L_2(P_0)$ convergence condition required by
Theorem~\ref{thm:cvtmle_representation_appendix}, namely,
$
\max_{j\in\{1,\ldots,J\}}
P_0
\left[
\left(
D_{n,j}^{\mathrm{prag},*}
-
D_{\infty}^{\mathrm{prag}}
\right)^{2}
\right]
=
o_p(1).
$

Finally, for any fixed pair of treatment rules
$(d_1^v,d_2^{v,r})$, define the second-order
remainder associated with the pragmatic-effectiveness estimand, relative to the true distribution $P_0$, by
\begin{align}
\operatorname{Rem}^{\mathrm{prag}}_{d_1^v,d_2^{v,r}}(Q,g)
&:=
\psi^{\mathrm{prag}}_{d_1^v,d_2^{v,r}}(Q)
-
\psi^{\mathrm{prag}}_{d_1^v,d_2^{v,r}}(Q_0)
+
P_0
D^{\mathrm{prag}}_{d_1^v,d_2^{v,r}}(Q,g).
\end{align}
% Finally, for any fixed pair of treatment rules
% $(d_1^v,d_2^{v,r})$, the second-order remainder for the pragmatic effectiveness is
For this estimand, the remainder can be written explicitly as
\begin{align}
\operatorname{Rem}^{\mathrm{prag}}_{d_1^v,d_2^{v,r}}(Q,g)
&=
P_0
\left[
\left\{
\frac{
g_0\left(d_2^{v,r}(V,R)\mid W,R\right)-g\left(d_2^{v,r}(V,R)\mid W,R\right)
}{
g\left(d_2^{v,r}(V,R)\mid W,R\right)
}
\right\}
\right.
\nonumber\\[-0.25em]
&\qquad\left.
\times
\left\{
\bQ_0\left(d_2^{v,r}(V,R),W,R\right)
-
\bQ\left(d_2^{v,r}(V,R),W,R\right)
\right\}
\right]
\nonumber\\
&\quad-
P_0
\left[
\left\{
\frac{
g_0\left(d_1^v(V)\mid W,R\right)-g\left(d_1^v(V)\mid W,R\right)
}{
g\left(d_1^v(V)\mid W,R\right)
}
\right\}
\right.
\nonumber\\[-0.25em]
&\qquad\left.
\times
\left\{
\bQ_0\left(d_1^v(V),W,R\right)
-
\bQ\left(d_1^v(V),W,R\right)
\right\}
\right].
\label{eq:cvtmle_remainder_explicit}
\end{align}
In particular,
$
\operatorname{Rem}_{n,j}^{\mathrm{prag}}
=
\operatorname{Rem}_{\donevnminusj,\dtwovrnminusj}^{\mathrm{prag}}
\left(
Q_{n,j}^{*},
g_{n,j}
\right)$.
Under Assumption~\ref{assump:cvtmle_positivity}, the
Cauchy--Schwarz inequality therefore gives, for some finite constant
$C<\infty$,
\[
\left|
\operatorname{Rem}_{n,j}^{\mathrm{prag}}
\right|
\leq
C
\left\|
\bQ_{n,j}^{*}-\bQ_0
\right\|_{P_0,2}
\left\|
g_{n,j}-g_0
\right\|_{P_0,2}.
\]
Therefore, it follows from
Assumption~\ref{assump:cvtmle_second_order} that
$
\frac{1}{J}
\sum_{j=1}^{J}
\operatorname{Rem}_{n,j}^{\mathrm{prag}}
=
o_p
\left(
n^{-1/2}
\right).
$

Since the targeting step satisfies
$
\frac{1}{J}
\sum_{j=1}^{J}
P_{n,j}
D_{n,j}^{\mathrm{prag},*}
=
0$,
Theorem~\ref{thm:cvtmle_representation_appendix} consequently yields
\[
\psi_{n,\mathrm{cv-tmle}}^{\mathrm{prag},*}
-
\tilde{\psi}_{n}^{\mathrm{prag}}(Q_0)
=
(P_n-P_0)
D_{\infty}^{\mathrm{prag}}
+
o_p
\left(
n^{-1/2}
\right).
\]

Hence, by the central limit theorem,
$
\sqrt{n}
\left(
P_n-P_0
\right)
D_{\infty}^{\mathrm{prag}}
\dto
N
\left(
0,\sigma_{\infty}^{2}
\right)$, where 
$
\sigma_{\infty}^{2}
=
\Var_{P_0}
\left\{
D_{\infty}^{\mathrm{prag}}(O)
\right\}.
$
\end{proof}

\end{document}